\documentclass{article}[12pt]
\usepackage{amsmath}
\usepackage{amsfonts}

\usepackage{wasysym}

\usepackage{color}

\advance \topmargin by -\headheight
\advance \topmargin by -\headsep
     
\evensidemargin \oddsidemargin
   \newtheorem{theorem}{Theorem}[section]
   \newtheorem{lemma}[theorem]{Lemma}
   \newtheorem{proposition}[theorem]{Proposition}
   \newtheorem{corollary}[theorem]{Corollary}
   
   \newtheorem{defn}[theorem]{Definition}

   \newenvironment{proof}[1][Proof]{\begin{trivlist}
   \item[\hskip \labelsep {\bfseries #1}]}{\end{trivlist}}
   \newenvironment{definition}[1][Definition]{\begin{trivlist}
   \item[\hskip \labelsep {\bfseries #1}]}{\end{trivlist}}

\newcommand{\beq}{\begin{equation}} 
\newcommand{\eeq}{\end{equation}}
\newcommand{\beqa}{\begin{eqnarray}}
\newcommand{\eeqa}{\end{eqnarray}}

\newcommand{\bthm}{\begin{theorem}}
\newcommand{\ethm}{\end{theorem}}

\newcommand{\bco}{\begin{corollary}}
\newcommand{\eco}{\end{corollary}}

\newcommand{\ble}{\begin{lemma}}
\newcommand{\ele}{\end{lemma}}

\newcommand{\bpr}{\begin{proposition}}
\newcommand{\epr}{\end{proposition}}

\newcommand{\bpf}{\begin{proof}}
\newcommand{\epf}{\end{proof}}

   \numberwithin{equation}{section}

\newcommand{\bI}{{\bf I}}
\newcommand{\bJ}{{\bf J}}

\newcommand{\bbC}{\mathbb{C}}

\newcommand{\bbE}{\mathbb{E}}

\newcommand{\bbN}{\mathbb{N}}

\newcommand{\bbR}{\mathbb{R}}

\newcommand{\bbZ}{\mathbb{Z}}

\newcommand{\cA}{\mathcal{A}}
\newcommand{\cB}{\mathcal{B}}

\newcommand{\cE}{\mathcal{E}}
\newcommand{\cF}{\mathcal{F}}

\newcommand{\cH}{\mathcal{H}}

\newcommand{\cR}{\mathcal{R}}

\newcommand{\fh}{\mathfrak{h}}

 \newcommand{\ol}{\overline}
 \newcommand{\wt}{\widetilde}

 \newcommand{\tn}{\textnormal}

\DeclareMathOperator{\id}{I}
\DeclareMathOperator{\de}{d}

\DeclareMathOperator{\tr}{{Tr}}

\begin{document}

\title{Log-dimensional bounds for the  spectral measure of the disordered Holstein model}
\author{Rajinder Mavi}
\maketitle 
%\begin{flushleft} 

\section{Introduction}
We will discuss the disordered Holstein model, defined by the Hamiltonian
  \begin{align}\label{holstein1}
    H & = (V + \gamma J) \otimes \id  +\alpha \sum_{u} \delta_u \otimes ( b_u + b^\dagger_u ) + \omega \id \otimes (\sum_{u} b_u^\dagger b_u ) + \frac{\alpha^2}{\omega}  
   \end{align}
   where $V + \gamma J$ is the standard discrete Anderson model and the operators $b_u,b_u^\dagger$ for $u \in \bbZ^D$ are standard annihilation and creation operators for a field of non-coupled quantum harmonic oscillators. The model and corresponding Hilbert space is discussed in further detail below in Section \ref{model}. 
   
   Briefly, the model describes a mobile tracer particle in a disordered environment. The tracer particle locally deforms the field of oscillators via the coupling term. This deformation allows the tracer particle to create and destroy excitations at it's current location. Thus the model does not conserve excitation (particle) number, which contributes to the difficulty in managing the entropy of the model. 
     
Recently Mavi and Schenker \cite{MS18L} demonstrated dynamical localization properties of the tracer particle for finite energies in the spectrum. The proof proceeds from a modification of the standard fractional moment method \cite{AW2015R}. The derived bounds, stated in this paper in Theorem \ref{greenfdecay}, decay with respect to the position of the tracer particle and the excitation field. However, this model defies traditional proofs of dynamical localization with the same bounds as randomness is only present at the tracer particle.  As the projection to a single tracer particle position is infinite dimensional, the tracer particle localization  does not imply a pure point spectrum.  Indeed it is possible that the disordered Holstein model exhibits resonant delocalization as sending $\gamma \to 0$ leads to degenerate eigenvalues. The degenerate eigenvalues correspond to reconfigurations of the excitation field. The existence of such resonant delocalization was established by Mavi and Schenker \cite{MS19R} in a simple model which would correspond to a single excitation permitted to occupy one of just two positions. 

In this paper we show log-dimensional properties of the spectrum of the Holstein model at finite energies for sufficiently small $\gamma > 0$. The existence of log-dimensional measures was previously discussed for one dimensional quasiperiodic Schrodinger operators in \cite{DL03L}. It should be noted that our results do not rule out pure point spectrum, but do rule out spectra `more continuous' than log-dimensional.

 \subsection{Model}\label{model}
We study the disordered Holstein model on a lattice $\Lambda \subset \bbZ^d$.  The standard model for a particle in a disordered system is the Anderson model, which we will write as
\[ h_{And} = \gamma J + V  \]
acting on $\fh_\Lambda = \ell^2(\Lambda)$. Here $V$ is a diagonal operator with random disorder, for $\phi \in \fh_\Lambda $  
\[ V\phi(u) = V_u\phi(u)  \]
where each term $V_u$ is iid on $[0,V_+ ]$ with bounded distribution $\nu$. The first term $J$ is the Laplacian operator defined, for $u,v\in\Lambda$ by
\[ \langle v | J |u\rangle  = 2d \delta_{u-v} - \delta_{|u-v| -1 }.  \] 
The Holstein model diverges from the Anderson model by affixing quantum harmonic oscillators to the lattice. For each site $u\in \Lambda$ we will write $b_u$ and $b_v^\dagger$ for the standard Bosonic annihilation and creation operators at $u$ and $v$ respectively, obeying commutation relations
\[  [b_u,b_v^\dagger]  = \delta_{u,v} 
  \tn{ and } [b_u^\dagger,b_v^\dagger] = [b_u,b_v] = 0.   \] 
 We will write $|\Omega\rangle $ for the unique vacuum state with no excitations.
 Let us define the spaces of functions
\[ \cA_\Lambda^N = \left\{ \xi : \Lambda \to \bbN:  N_\xi < N  \right\} . \]
where  $N_\xi := \|\xi\|_1 $. For $\xi \in \cA_\Lambda^N$, define
\[  |\xi \rangle := \prod_{ u \in \Lambda } \frac{(b_u^\dagger)^{\xi_u}}{\sqrt{\xi_u!}}|\Omega\rangle   \]
  the functions defining configurations of oscillators. When $N = \infty$ this permits configurations with any finite number of oscilations and forms the basis elements for the Fock space 
\[   \cF_\Lambda = \left\{  \sum_{\xi\in \cA_\Lambda^\infty} f_\xi|\xi\rangle   : \sum_{\xi \in \cA_\Lambda^\infty}(1+ N_\xi )|f_\xi |^2 < \infty \right\}. \]
Note $ \cF_\Lambda$ defines the Hilbert space for the   occupancy counting operator 
\[ h_{osc} =  \sum_{u\in\Lambda} b_u^\dagger b_u.  \]
The Hilbert space upon which the Hamiltonian \eqref{holstein1}  is defined may now be written as $\cH = \ell^2(\bbZ^d) \otimes \cF_{\bbZ^d}$. The second term in \eqref{holstein1}  forms the interaction between the mobile tracer particle and the field of oscillators, the final term is merely  a shift which moves the bottom of the spectrum to 0. Setting $J\otimes \id = \Delta$ we may recast the Hamiltonian as 
  \begin{align}\label{holstein2}
    H  
     & = \gamma \Delta + H_f .
   \end{align}
 The convenience of introducing this field operator $H_f$ is that we can easily write down the associated eigenfunctions with the aid of the Glauber displacement operator, discussed in detail in \cite{MS18L}. We label the set of eigenfunctions as $\cE$ which, for $u\in \bbZ^d$ and $\xi \in \cA_{\bbZ^d}^\infty$, are defined as 
  \begin{align} \label{eigenbasis} |u,\xi \rangle = |u\rangle \otimes e^{\beta(b_u^\dagger - b_u) } | \xi\rangle  \end{align} 
  where $\beta = \alpha / \omega$. Note the displacement operator only acts on the field at the site $u$, the location of the tracer particle.
  For any  $|u,\xi\rangle \in \cE $  we have 
  \[  H_f | u,\xi \rangle = (V_u + \omega N_\xi)|u,\xi\rangle. \]  
   The spectrum of $H_f$ is exactly
  \[ \sigma(H_f) = \overline { \cup_{n} \bI_{n}} \tn{ where, for each $n\in \bbN$, } \bI_{n} 
      =(n \omega,n\omega+V_+) , \]
  when $\omega > V_+$, the spectrum is divided into bands of width $ V_+$.   The $n^{th}$ spectral band, $\bI_{n}$, corresponds to the single tracer particle with $n$ excitations in the oscillator lattice. 
  For the full Hamiltonian $H$, the separation of spectral bands is preserved for small $\gamma$
    \begin{align} \label{spec H} 
    \sigma(H) =  \overline{\cup_{n} \bI_{n,\gamma}} \tn{ where, for each $n\in \bbN$, } \bI_{n,\gamma} 
      =(n \omega,n\omega+V_+ + 4d\gamma).  \end{align}
  Note that the field Hamiltonian $H_f$ does not commute with the kinetic term $\Delta $, thus excitation number is not preserved under dynamical evolution.

 The following result, which states dynamical localization with respect to the localization of the tracer particle, holds for elements of $\cE$.
 \begin{theorem}[Theorem 1.6 from \cite{MS18L}]
 \label{thm MS1} Let $\omega > V_\infty$ and $I  = (-\infty , \omega (N+1))$ for $N\in \bbN $.
 
    Fix $\lambda > 0$, then there is $\wt\gamma>0$ so that, for $0 < \gamma < \wt \gamma$,
    \[  \bbE \left( \sup_{f \in \cB_1(I)} |\langle v,\zeta | f(H) |u,\xi \rangle| \right) \leq  \frac{C_N}{1 + \|\zeta\|^{\frac12 - \epsilon}} e^{- \lambda \| u - v\|} \]
   for any $\epsilon > 0 $ and  $ \langle v,\zeta|,\langle u,\xi | \in \cE$.
 \end{theorem}
In \cite{MS18L} it is important that   $\gamma $ is taken small enough that the spectral bands $\ol\bI_{n,\gamma}$ are kept separated.

  \subsection{Main theorem}

   First we introduce the concept of the Hausdorff measure of sets.
  \begin{defn}\label{def gauge}\  
  
  We denote as a Hausdorff function those functions which are continuous and monotonically increasing mapping $[0,\infty)$ to $[0,\infty]$ such that $h(0) = 0$.
  
  Given a set $S\subset \bbR$ and $\delta> 0$, a  $\delta$-cover of $S$ is a countable collection of intervals $U_i$ such that $S \subset \cup_i U_i$ and  $|U_i| < \delta$ for all $i$.
  
 Given a set $S \subset \bbR $, the  $\rho$-Hausdorff measure of $S$ is defined as
\[  \mu^{\rho}(S) = \lim_{\delta \to 0} 
      \inf_{\delta\tn{-covers of }S}  \sum_i \rho( |U_i| ) .  \] 
  \end{defn}
A spectral measure is called $\rho$-Hausdorff singular if it is supported on a set of zero $\mu^\rho$ measure.
 For any $N \geq 0$ let us define  
 \begin{align} \label{N bands open}
 \bJ_{N} = \cup_{n=0}^N \bI_{n,\gamma}, \end{align}
 along with the corresponding spectral projection $\chi_{\bJ_N}$.  
 \bthm\label{main thm}
   For any $N \geq 0$, there is sufficiently small $\gamma$ such that, for any $|u,\xi\rangle \in \cE $ and $\psi = \chi_{\bJ_N} |u,\xi\rangle $, the spectral measure $\mu_{ \psi } $ of  $H$ is purely $\rho$-Hausdorff singular for  
   \[ \rho(s) = |\log(s)|^{-p} \]
   for any $p>0$.
 \ethm

 The remainder of this paper is organized as follows. In Section \ref{Hcqd}, we discuss the necessary spectral measure theory in two parts. In Section \ref{DBm} we further discuss the Hausdorff continuity properties of Borel measures. Next, in Section \ref{QdHcm} we relate spectral continuity to the quantum dynamics of $H$. In the final section we study upper and lower bounds of quantum dynamics of $H$. In Section  \ref{Lbdwp} we use the continuity of Hausdorff measures to imply the spread of the wave packet, in Section \ref{Ubdwp} we utilize upper bounds on quantum dynamics of $H$ to localize the wave packet.  Following this, we conclude the paper with the proof of Theorem \ref{main thm}.

 \section{Hausdorff continuity of measures and quantum dynamics} \label{Hcqd}
 We begin this  section with decompositions of  spectral measures. The  general decomposition of Borel measures is carried out with respect to continuity with respect to gauge functions introduced in Definition \ref{def gauge}. The continuity of the spectral measure  will bound quantum dynamics as discussed in Section \ref{QdHcm}.
   \subsection{Decompositions of Borel measures}\label{DBm}
  The decomposition of measures are determined by their local concentration properties as measured by a modulus of continuity. 
 \begin{defn}
 For a given Borel measure $\mu $ and gauge function $\rho$, we define the upper $\rho$-derivative as 
\[  \ol D_{\rho} \mu(x) = \lim_{\delta \to 0} \sup_{\substack{I \ni x \\ |I|< \delta} } \frac{\mu(I)}{\rho(I)}   \]
 \end{defn}
  Let $K_0,K_+,K_\infty$ be the sets of points $x$ of $\bbR$ such that $\ol D_{\rho} \mu(x)$ takes the value of 0, takes a finite positive value, or takes the value $+\infty$ respectively. Our utility for introducing $\ol D_\rho$ is in the following theorem.

\bthm[Theorem 67 of \cite{R98H}] \label{thm67} Given Borel measure $\mu$, let sets $K_0, K_+, K_\infty$ be defined as above. Then $K_0, K_+, K_\infty$ are Borel sets and the following holds
\begin{enumerate} 
  \item $\mu^\rho (K_\infty) = 0 $
  \item $K_+$ is $ \mu^\rho$ $\sigma $ - finite.
  \item $ \mu(E \cap K_+) = 0 $ if $\mu^\rho(E) = 0$.
  \item $ \mu(E \cap K_0) = 0 $ if $E $ is  $\mu^\rho $ $\sigma$ - finite.
\end{enumerate}

\ethm

By definition, $K_0 \sqcup K_+ \sqcup K_\infty = \bbR$, so we can write
$  d\mu = \chi_{K_0}d\mu +  \chi_{K_+}d\mu + \chi_{K_\infty} d\mu$, which defines a decomposition of measure.
\begin{defn}
  Suppose $\mu$ is a Borel measure and $\rho$ is a Hausdorff function.
\begin{itemize}
 \item[1.] $\mu$ is strongly $\rho$ continuous (S$\rho$C) if for any set $E$ of finite $\mu^\rho$ measure, $\mu (E) = 0$.
 \item[2.] $\mu$ is $\rho$ absolutely continuous ($\rho$AC) if there exists a function $f$ such that $d\mu(x) = f(x)d\mu^{\rho}(x) $.
 \item[3.] $\mu$ is $\rho$ singular ($\rho$S) if it is supported on a set $S$ such that $\mu^\rho(S) = 0$. 
 \end{itemize}
\end{defn}
The following corollary is a consequence of Theorem \ref{thm67} and the above definition.
\bco
 Given a Borel measure $\mu$ and a Hausdorff function $\rho$ there is a unique decomposition
\beq\label{meas decomp}
\mu = \mu_{S \rho C} + \mu_{\rho AC} + \mu_{\rho S}
\eeq
where $\mu_{S \rho C} $ is $S \rho C$, $\mu_{\rho AC}$ is $\rho AC$, and $\mu_{\rho S}$ is $\rho S$.
\eco
We say a Borel measure $\mu$ is $\rho$ continuous if $\mu_{\rho S} \equiv 0$. A stronger notion of $\rho$ continuity is useful in the following section.
 
\begin{definition}
A Borel measure $\mu$ is uniformly $\rho$ Holder (U$\rho$H) if there is a constant $C_\mu$ such that, for all $\epsilon < 1$,
\[   \sup_{s} \mu_{\psi}(s - \epsilon/2,s + \epsilon/2 ) \leq C_\mu \rho( \epsilon )  \]
  
\end{definition}
The following corollary is a further consequence of Theorem \ref{thm67}. 
\bco \label{Co UrhoH decomp}
Suppose any $\mu$ is a Borel measure which is $\rho$-continuous for some Hausdorff function $\rho$. Then for any $\epsilon > 0$, there is a U$\rho$H measure $\mu_1$ and a $\rho $-continuous measure  $\mu_2$  so that, $\mu = \mu_1 + \mu_2$ and $\mu_2(\bbR) < \epsilon $  
\eco

  \subsection{Quantum dynamics for Hausdorff continuous measures}\label{QdHcm}
  In this section we consider evolution of an initial state $\psi \in \cH$, a seperable Hilbert space, with respect to  a general Hamiltonian $H$.   For given self adjoint operator $A$ and vector $\psi \in \cH$, we will use the following notation
  \beq \label{A average}
    \langle A \rangle_{\psi,T} 
    := \int_0^T|\langle\psi(t)|A|\psi(t)\rangle\rangle | \frac{\de t}{T}
  \eeq
    for $\psi(t) = e^{-itH}$. Schrichartz \cite{S90F} and Last \cite{L96Q} have demonstrated a connection between quantum dynamics and spectral dimension for power-law $\rho$ which may be generalized to the following. (See also a similar result in \cite{lp})

\begin{proposition}  
Assume $\psi \in \cH $ is such that $\mu_\psi$ is U$\rho$H. Then there is a constant $C_\psi$ such that 
\[ \langle A \rangle_{\psi,T }
          \leq  C_\psi  \left[\rho\left(T^{-1} \right)\right]^{1/2}  \]
    for any rank one projection operator $A $.
\end{proposition} 
 \begin{proof}
   Let  $\cH_{|\psi} $ be $\cH$ restricted to the $H$ cyclic subspace generated by $\psi$. There is a unitary operator $U$ from $ \cH_{|\psi}$ to $L^2(\bbR, d\mu_\psi)$ such that $H$ is equivalent to multiplication by $x$. Let $P = |\phi\rangle \langle \phi |$, then there is $f_\phi$ so that $U : \phi \mapsto f_\phi $, thus
   \[  \langle \phi , e^{-iH t} \psi \rangle 
        =  \int_\bbR e^{-ixt} f_\phi(x) \de \mu_\psi(x).   \]
  We can insert this into the calculation for \eqref{A average} , to find
  \begin{align}
     \langle A \rangle_{\psi,T} 
     &  =  \int_0^T \left|\int_\bbR e^{-ixt} f_\phi(x) \de \mu_\psi(x) \right|^2 \frac{\de t}{T}
     \\ \nonumber 
     & =   \int_0^T  
     \int_\bbR \int_\bbR 
     e^{-i(x-y)t} f_\phi(x) \ol{f_\phi(y)} 
   \de \mu_\psi(y)\de \mu_{\psi}(x) \frac{\de t}{T}
     \\ \nonumber 
     & \leq  
     e \int_\bbR \int_\bbR f_\phi(x)\ol{f_\phi(y)} 
     \int_0^T  e^{-t^2/T^2 -i(x-y)t} \frac{\de t}{T}
     \de \mu_\psi(y) \de \mu_{\psi}(x)
     \\ \nonumber 
     & \leq  e \sqrt{\pi} 
     \int_\bbR \int_\bbR 
     e^{-T^2(x-y)^2/4}
      |f_\phi(x)|  |f_\phi(y)| 
     \de \mu_\psi(y)  \de \mu_{\psi}(x) 
  \end{align}
  Using Cauchy-Schwarz inequality twice, we have 
  \begin{align} \label{CS1}
   \langle A \rangle_{\psi,T} 
   & \leq e\sqrt{\pi} 
   \left( \int |f_\phi(y)|^2 \de \mu_\psi(y)\right)
    \left(\int   
  \int e^{-T^2(x-y)^2/2} 
  \de\mu_\psi(y)
  \de\mu_\psi(x) \right)^{1/2}
   \end{align}      
   For $T > 1$ we can partition $\bbR = \cup_{k\in\bbZ} 
    [x+\frac{k}{T} -\frac{1}{2T},x + \frac{k}{T}  + \frac{1}{2T})$ and bound
   \[ \int e^{-T^2(x-y)^2/2} \de\mu_\psi(y)
    \leq 2 \sum_{k\geq 0} e^{- k^2/2 }\rho\left(\frac{1}{T}\right) 
    \leq C \rho\left(\frac{1}{T}\right) 
    \]    
  Applying this to \eqref{CS1} obtains     
   \[ \langle A \rangle_{\psi,T} 
   \leq e\sqrt{\pi} C \|\phi\|^2 \|\psi\| \left[\rho(T^{-1}) \right]^{1/2} 
   \]     
     which completes the theorem.   
        
 \end{proof}

\bco\label{Co UrhoH}
  If $\psi$ is such that $\mu_{\psi}$ is U$\rho$H
 then for any compact operator $A$,
 \[  \langle A \rangle_{\psi,T} 
          \leq C_\psi^{1/p}  \|A\|_p \left[ \rho\left(\frac{1}{T} \right) \right]^{\frac{1}{2p}}  \]
 where $\|A\|_p = \left(\tr|A|^p \right)^{1/p}$ denotes the $p^{th}$-Schatten norm of $A$ for $p \geq 1$.
 \eco

 \section{Dynamical properties of $H$}
   
\subsection{Lower bounds on  diffusion of the wave packet} \label{Lbdwp}

For $L\geq 1 $  define the truncation of the lattice 
\[ \Lambda_L = \{ u \in \bbZ^d : |u| < L \},   \]
projections for the position of the tracer particle and the positions of less than $K \geq 1$ excitations of the field within $\Lambda_L$ are
  \[  R_{L,K} = \sum_{u\in \Lambda_L} \sum_{ \xi \in \cA_{\Lambda_L}^K} |u,\xi\rangle \langle u,\xi|.  \] 
Setting $L$ or $K$ to $\infty$ projects to finite excitation numbers or spatial lattices respectively. For any $K,L > 0$,
\begin{align}\label{QR trace}
  \tr(R_{L,K})
  = 
  \sum_{k=0}^{K-1} \frac{|\Lambda_L|^{k+1}}{k!} \leq CL^{dK} 
\end{align}
as the equality is the count for distributing 1 tracer particle and less than $K$ indistinguishable excitations.

 For $\psi \in \cH$, spectral measure $\mu_\psi$ may be decomposed according to \eqref{meas decomp} into $\rho$ singular and continuous parts using the sets $K_{<\infty}(\mu_\psi) = K_0(\mu_\psi) +K_+(\mu_\psi) $ and $ K_{\infty}(\mu_\psi)$. In particular, for all basis elements $|u,\xi\rangle $ we can construct the set  \[ K_\infty = \cup_{u \in \bbZ^d} \cup_{\xi \in \cA^\infty_{\bbZ^d}} K_\infty(\mu_{|u,\xi\rangle}) \]
 so that we may define spectral projections $ P_{\rho S} = P_{K_\infty} $ and $P_{\rho C} = (1 - P_{\rho S}) $ which in turn define closed orthogonal subspaces $\cH_{\rho S} = P_{\rho S} \cH $ and $ \cH_{\rho C} = P_{\rho C} \cH $ such that $\cH = \cH_{\rho S} \oplus \cH_{\rho C}$.
           
           \bthm \label{thm deloc}
           Fix $q > 0$ and $p > 2 q d $.  If $\psi$ is such that $P_{\rho C}\psi \neq 0$ 
             \[\rho( s ) = |\log( s )|^{-p} \] then, for large enough $T$,
  \begin{align*}  
    \left\langle(\id - R_{L,K})  \right\rangle_{\psi,T}
            &\geq
           \frac12 \|\psi_{\rho C} \|^2  
           \end{align*}
            for any $ K = K_T$ such that $K_T\geq 1  $ and $\frac{K_T}{\log\log T}  \to 0$, and $ L = L_T = (\log T)^{q/K_T}$.
             
           \ethm

\bpf We will decompose $\psi$ into a uniformly $\rho$ H\"older portion and a remainder.
 First write $\psi $ as a $\rho $ singular plus $\rho$ continuous sum
\[  \psi = \psi_{\rho C} + \psi_{\rho S}.\]
Further, given $\epsilon > 0$, we may use Corollary \ref{Co UrhoH decomp} to write the decomposition  
\[  \psi_{\rho C} = \psi_{1} +\psi_{2}  \]
such that  $\|\psi_2\| \leq   \epsilon $ and $\mu_{\psi_1}$ is U$\rho$H. Notice, by Corollary \ref{Co UrhoH}  and \eqref{QR trace}
\beq \label{ineq Q} \langle  R_{L,K} \rangle_{\psi_1,T}  
      < C_{\psi_1}   \tr( R_{L_T,K_T} )  \rho(T^{-1}) \leq C_{\psi_1}' \frac{  L_T^{dK_T} }{ (\log T)^{p/2} } 
      = C_{\psi_1} (\log T)^{qd - p/2}   \eeq 
      which approaches 0 for large $T$.
We apply this bound to the calculation of the moment of the original initial state $\psi$
\begin{align}
 \left\langle  R_{L,K}   \right\rangle_{\psi,T}
 & =   \label{ineq seq}
   \int_0^T\|  R_{L,K}  \psi(t)\|^2 \frac{dt}{T}  \\
   & \leq \nonumber 
    \int_0^T \left(\|  R_{L,K}  \psi_1(t)\| + \|\psi_2\| +\|\psi_{\rho S}\| \right)^2\frac{dt}{T}  \\
 & \leq \nonumber
  \left( \left(\left\langle R_{L,K}   \right\rangle_{\psi_1,T}\right)^{1/2} + \|\psi_{2} \|+\|\psi_{\rho S}\| \right)^2 
 \\&\leq \nonumber
   \left(  2 \epsilon +\|\psi_{\rho S}\| \right)^2 
\end{align}
where Minkowski's inequality is used on the third line and \eqref{ineq Q} is used in the final step. 
  Then it follows that
  \begin{align*} 
      \left\langle(1 - R_{L,K})  \right\rangle_{\psi,T}
            & = \|\psi\|^2 - \left\langle( R_{L,K})  \right\rangle_{\psi,T} 
  \geq \frac12 \|\psi_{\rho C}\|^2
  \end{align*} 
  for sufficiently large $T$.
\epf

  \subsection{Upper bounds on diffusion of the wave packet} 
  
  \label{Ubdwp}
  
  We will require a psuedo-metric to state the decay bound of the expectation of the fractional Green's function. 
  First we introduce a function on $ \bbZ^d \times \cA_{\bbZ^d}^{\infty 2} $ let
  \[  \cR_{\xi|\zeta}(u) = \max\{ \|v-u\| :\zeta(v) > 0 ;\zeta(v)\neq \xi(v) \}, \]
  we may now define 
  \[ \Upsilon(v,\zeta; u,\xi ) = \max \{\|u-v\| , \cR_{\xi|\zeta}(u) , \cR_{\zeta|\xi}(v) \}, \]
  that $\Upsilon$ is a pseudo-metric is demonstrated in \cite{MS18L}.
  First we state the fractional moment decay of the Green's function.
  
  \begin{theorem}[Corollary 1.5 from \cite{MS18L}]\label{greenfdecay}
     Suppose $\nu$ is supported on $[0,V_+]$ and that $\omega > V_+$. Fix $s < 1 $ and $\lambda > 0$. Then for any $N \geq 0$ there is $\gamma_N >0$ so that for any $ 0 < \gamma < \gamma_N $ there is finite $ C $ so that, for any $z \in \bbC\setminus \bbR$ so that $\Re (z) \in S_0\cup S_1 \cup \cdots \cup S_N $  we have
     \[ \bbE( |\langle u, \xi| ( H - z)^{-1}| v,\zeta \rangle  |^s ) < C e^{- \lambda \left( \Upsilon(u,\xi;v,\zeta) +  \left|\sqrt{N_\xi} - \sqrt{N_\zeta}\right| \right) } \]
     for any $|u,\xi\rangle,|v,\zeta\rangle \in\cE$  
  \end{theorem}
  From this inequality we shall show dynamical localization for small times on the spectral set $\bJ_N$ defined in  \eqref{N bands open} 
  \begin{theorem}\label{Thm T Avg}
     Fix $N\geq 0$, $\lambda > 0$, 
     and $ 0 < s < 1$. Then there is $\gamma_N >0$ so that for any $ 0 < \gamma < \gamma_N $ so that for every $|v,\zeta\rangle \in \cE $ there is, almost surely, a finite $ C $ so that
      \begin{align} \label{T Avg bound}
        \frac 1T \int_0^T |\langle v, \zeta| e^{itH} \chi_{\bJ_N} | u,\xi \rangle  |^2 \de t <  C_{s,N,\lambda;|u,\xi\rangle } T^{1-s} e^{- \lambda \left( \Upsilon(u,\xi;v,\zeta) +  \left|\sqrt{N_\xi} - \sqrt{N_\zeta}\right| \right) }  \end{align}
     for any $|v,\zeta\rangle \in\cE$.  
  \end{theorem}
  
  \begin{proof} 
  Let  $ \gamma_N$ be as defined in Theorem \ref{greenfdecay}. Let $0 < \gamma < \gamma_N$, $c_{\gamma} =  V_+ + 4d\gamma$, and let $0 < \epsilon <  \frac12(\omega - c_{\gamma}) $.  
  For $0 < \epsilon,\delta< 1$ and $0 \leq k \leq N$,  let $ I_{k,\gamma,\epsilon} = (\omega n+\epsilon,\omega n + c_\gamma - \epsilon) $, and define  the contour in the complex plane
  \[\Gamma_{n,\gamma,\delta,\epsilon} = \partial\{z =x+iy \in\bbC:  x\in  I_{n,\gamma,\epsilon}  ;| y |  < \delta\} \]
   oriented in the counterclockwise direction.
  We now consider the averages
  \begin{align}
     \bbE\left[ \frac{1}{T} \int_0^T \left|\langle u,\xi |  e^{-it H } \chi_{I_{n,\gamma,\epsilon}} |v, \zeta \rangle\right| \de t\right] 
    & = 
     \bbE\left[ \frac{1}{T} \int_0^T \left|\langle u,\xi | \left( \frac{1}{2\pi i}\int_{\Gamma_{n,\gamma,\delta,\epsilon}} \frac{ e^{-itw} }{ w-H } \de w \right) |v, \zeta \rangle\right| \de t\right] ,\end{align}
     for $\chi_{I_{n,\gamma,\epsilon}}  $ the indicator function for $ I_{n,\gamma,\epsilon}  $. Although the contour of integration passes through the spectrum, given Theorem \ref{greenfdecay} it is not hard to see the integrals converge.
  For any $\omega = x+iy \in S_k $, with $|y|\leq \delta$, and  for any  $0 < s < 1 $ 
  \begin{align}
  \bbE \left[ \frac{1}{T} \int_0^T \left|\langle u,\xi |  \frac{ e^{-itw} }{ w-H }   |v, \zeta \rangle\right| \de t\right] 
  & =  
   \frac{e^{ T \delta }}{T |y|^{1-s} } \int_0^T \bbE \left[ \left|\langle u,\xi |  \frac{ 1 }{ w-H }   |v, \zeta \rangle\right|^s
    \right] \de t 
    \\ \nonumber
  & \leq  
   \frac{e^{ T \delta }}{ |y|^{1-s} } C e^{- \lambda \left( \Upsilon(u,\xi;v,\zeta) +  \left|\sqrt{N_\xi} - \sqrt{N_\zeta}\right| \right) } 
  \end{align}
 where the inequality follows from Theorem \ref{greenfdecay}. For any $0 < s < 1$ there is a $C_1$, so that,  integrating about the contour,
 \begin{align}
  \frac{1}{2\pi} \int_{\Gamma_{n,\gamma,\delta,\epsilon}} \bbE \left[ \frac{1}{T} \int_0^T \left|\langle u,\xi |  \frac{ e^{-itw} }{ w-H }   |v, \zeta \rangle\right| \de t\right]  \de w   
  & \leq  
   \frac{e^{ T \delta }}{ \delta^{1-s} } C_1 e^{- \lambda \left( \Upsilon(u,\xi;v,\zeta) +  \left|\sqrt{N_\xi} - \sqrt{N_\zeta}\right| \right) }. 
 \end{align}
 Setting $\delta = T^{-1}$, and applying Fubini's theorem  
  \begin{align}
     \bbE\left[ \frac{1}{T} \int_0^T \left|\langle u,\xi | \left( \frac{1}{2\pi i}\int_{\Gamma_{n,\gamma,\delta,\epsilon}} \frac{ e^{-itw} }{ w-H } \de w \right) |v, \zeta \rangle\right| \de t\right] 
     \leq 
    C_2  T^{1-s} e^{- \lambda \left( \Upsilon(u,\xi;v,\zeta) +  \left|\sqrt{N_\xi} - \sqrt{N_\zeta}\right| \right) }. 
    \end{align}
   A simple Borel-Cantelli lemma then implies, for any $|v,\zeta\rangle$ and $0< \lambda'<\lambda$, there is $C_3$ so that 
  \begin{align}
      \frac{1}{T} \int_0^T \left|\langle u,\xi | e^{-it H } \chi_{I_{n,\gamma,\epsilon}}  |v, \zeta \rangle\right| \de t  
     \leq 
    C_3  T^{1-s} e^{- \lambda' \left( \Upsilon(u,\xi;v,\zeta) +  \left|\sqrt{N_\xi} - \sqrt{N_\zeta}\right| \right) } 
    \end{align}
    almost surely.  
    As the integrand is bounded by 1 we have 
      \begin{align}
      \frac{1}{T} \int_0^T \left|\langle u,\xi | e^{-it H } \chi_{I_{n,\gamma,\epsilon}}  |v, \zeta \rangle\right|^2 \de t  
     \leq 
    C_3  T^{1-s} e^{- \lambda' \left( \Upsilon(u,\xi;v,\zeta) +  \left|\sqrt{N_\xi} - \sqrt{N_\zeta}\right| \right) } ,
    \end{align}
     taking $\epsilon \to 0$ completes the proof.
   \end{proof}

   Now we move onto the upper bound for diffusion of the wavepacket. Let $K_T$ and $ L_T$ be  functions approaching $\infty$ as $T\to \infty$ such that   $ K_T < L_T^{1-\epsilon}$.
  
   \begin{theorem}
   \label{thm loc}
 Given $N$ there is $\gamma_N> 0$  sufficiently small so that, 
 for any $0 < \gamma < \gamma_N $ the following holds almost surely.
    For any $ |u,\xi\rangle $ and all sufficiently large $T$,
   \begin{align}
    \langle ( \id - R_{L_T,K_T})\rangle_{\psi,T} & \leq \frac{ (N+1) \omega }{K_T}  
    +C_{u,\xi}  e^{- \lambda   L_T }
   \end{align} 
   where $ \psi = \chi_{\bJ_N}|u,\xi\rangle $. 
   \end{theorem}
   
  \begin{proof} 
    
    Let us again assume  $L_0 \in \bbN$ is large enough that  $|u,\xi \rangle \in Q_{L_0}R_{L_0,N_\xi+1}\cE$.
    We will apply an energy argument to bound large excitations, 
    %for $\phi \in \ran(\chi_{\bJ_N}) $ 
  % \begin{align}\label{e3}
  % \| ( I  - R_{\infty,K_T}  ) \phi\|^2 (K_T-4d\gamma) \leq \langle (1- R_{\infty,K_T}) \phi,H   ( I  - R_{\infty,K_T}  )   \phi \rangle & \leq  \langle  \phi,H     \phi \rangle\leq N\omega + V_+ + 4d\omega
  % \end{align}	 
    \begin{align} \label{I-RooK}
        \langle(\id-R_{\infty,K})\rangle_{\psi,T}   & \leq \frac{ (N+1) \omega }{ K }
    \end{align}
  Considering the bound in \eqref{T Avg bound} let us introduce the notation  
     \[ F_\lambda(v,\zeta)  =   e^{- \lambda  \Upsilon(u,\xi;v,\zeta)  }, \]
  and the subsets  of $\cE$,   
     \[   \cE_{L,K} := (R_{L+1,K} - R_{L,K} ) \cE.   \]
     For $ K < L_T^{1-\epsilon}   $ and sufficiently large $L_T $
     \begin{align} \label{F sum} 
    \sum_{L \geq L_T}  \sum_{ |v,\zeta\rangle \in \cE_{L,K}  } F_\lambda(v,\zeta) 
      &\leq
     C   e^{\gamma L_0 } 
        \sum_{L \geq L_T} 
     L^{Kd}  e^{- \lambda L  }  
     \leq    
    C_{\epsilon'}  
     e^{- (\lambda -\epsilon') (L_T - L_0) }.
      \end{align}
  Combining Theorem \ref{Thm T Avg}  and \eqref{F sum} we have   
    \begin{align} 
    \label{RooK-RLTK}
      \langle (R_{\infty, K} - R_{ L_T,K}) \rangle_{\psi,T}
      = \sum_{L > L_T}  \langle( R_{L+1,K}  - R_{L,K}  )\rangle_{\psi,T} & \leq 
    C_{\epsilon'}  
     e^{- (\lambda -\epsilon') (L_T - L_0) }.
    \end{align}
    Combining \eqref{I-RooK} and
     \eqref{RooK-RLTK} completes the proof.

\end{proof}

 The proof of Theorem \ref{main thm} is a combination of theorems \ref{thm deloc} and \ref{thm loc}. Suppose $|u,\xi\rangle \in \cE $ is such that $P_{\rho C} \psi \neq 0$. 
    Let $K_T =(\log\log T)^{1/2} $ and $ L_T = (\log T)^{ q / K_T} = e^{q K_T}$  for any $0 < q < p/(2d)$.
    Then Theorem \ref{thm deloc} implies 
 \[    \left\langle(\id - R_{L_T,K_T})  \right\rangle_{\psi,T}
            \geq
           \frac12 \|\psi_{\rho C} \|^2  \]
  On the other hand, Theorem \ref{thm loc} implies $ \left\langle(\id - R_{L_T,K_T})  \right\rangle_{\psi,T} \to 0 $ as $T \to \infty$.
 Thus $\|\psi_{\rho C} \| = 0 $, which implies $\psi $ is purely $\rho $-Hausdorff singular.
  
\bibliographystyle{plain}	
\bibliography{zdimrefs}

%\end{flushleft}
\end{document}